%% file: main.tex
\documentclass{article}

\usepackage{graphicx} 
\usepackage[ruled, lined, linesnumbered, longend]{algorithm2e}
\usepackage[utf8]{inputenc}
\usepackage{indentfirst}
\usepackage{nicefrac}
\usepackage{amsbsy,amssymb,amsthm,amsmath, amsfonts,hyperref,fixltx2e,bbding} 
\usepackage{xcolor}
\usepackage{authblk}

\usepackage{hyperref}
\usepackage{url}

\usepackage{hyperref}
\usepackage{graphicx}
\usepackage{ragged2e}

\newtheorem{theorem}{Theorem}
\newtheorem{lemma}{Lemma}
\newtheorem{example}{Example}
\newtheorem{corollary}{Corollary}
\newtheorem{prop}{Proposition}

\SetKw{KwBy}{by}

\newcommand{\vb}{\mathbf{v}}
\newcommand{\ab}{\mathbf{a}}

\newcommand{\X}{{\bf X}}
\newcommand{\U}{{\bf U}}

\input{comm-anthony}

\title{Logical Languages Accepted by \\ 
Transformer Encoders with Hard Attention}

\author[1]{Pablo Barcel\'o}
\author[1]{Alexander Kozachinskiy}
\author[2]{Anthony Widjaja Lin}
\author[3]{Vladimir Podolskii}
\affil[1]{Institute for Mathematical and Computational Engineering,  Universidad Cat\'olica de Chile \& IMFD Chile \& CENIA}
\affil[2]{TU Kaiserslautern, Kaiserslautern, Germany \&
Max Planck Institute for Software Systems, Kaiserslautern, Germany}
\affil[3]{Courant Institute of Mathematical Sciences, NewYorkUniversity, NY, USA \& Steklov Mathematical 
Institute of Russian Academy of Sciences, Moscow, Russia}

\date{} 

\begin{document}

\maketitle

\input{abstract}



\section{Introduction}

Transformers have revolutionized natural language processing by facilitating the efficient and effective modeling of intricate contextual relationships within text \cite{DBLP:conf/nips/VaswaniSPUJGKP17}. 
This remarkable capability has sparked numerous investigations into the potential boundaries of transformers' power \cite{DBLP:journals/tacl/Hahn20,DBLP:conf/acl/YaoPPN20,DBLP:journals/jmlr/PerezBM21,DBLP:conf/icml/WeissGY21,DBLP:journals/tacl/HaoAF22,DBLP:conf/acl/0001C22, DBLP:conf/emnlp/BhattamishraAG20,DBLP:conf/icml/0001CP23}. One natural method for addressing this question is to explore the classes of formal languages that these architectures can recognize. This approach provides an insight into their strengths and limitations. The response to this question naturally relies on the specific features allowed within transformer encoders. These encompass the interplay between encoders and decoders, the kind of functions used for positional encodings and attention mechanisms, and considerations of fixed or unbounded precision, among other factors.

While the capacity of transformers that incorporate both encoders and decoders to recognize languages is well understood today (indeed, such architectures are Turing-complete and can thus recognize any computable language \cite{DBLP:journals/jmlr/PerezBM21}), the expressive power of transformer encoders has not been fully elucidated to date. 
\emph{Unique Hard Attention Transformers (UHAT)} are a class of transformer encoders that has been a subject of many recent papers.
As was shown 
by \cite{DBLP:journals/tacl/HaoAF22},  UHATs
recognize only languages in 
${\sf AC}^0$, i.e., 
recognized by families of Boolean circuits of unbounded fan-in that have constant depth and polynomial size. Intuitively, this means that UHATs are rather weak at ``counting'' (more precisely, reasoning about the number of occurrences of various letters in the input word). For example, consider 
the following two languages: {\em majority} and {\em parity}. The first one corresponds to the set of words over alphabet $\{a,b\}$ for which the majority of positions are labeled by $a$, while the second  checks if the number of positions labeled $a$ is even. That these languages are not in ${\sf AC}^0$ follows from a groundbreaking result in circuit complexity theory \cite{DBLP:conf/focs/FurstSS81,DBLP:journals/apal/Ajtai83}). Hence, 
they are neither accepted by UHATs. 
However, which fragment of the ${\sf AC}^0$ languages can actually be recognized by  UHATs remains an unresolved question. 

We start by showing that not all ${\sf AC}^0$ languages can be accepted by UHATs. This is obtained by combining results from  \cite{DBLP:journals/apal/Ajtai83} and  \cite{DBLP:journals/tacl/Hahn20}. Based on the previous observation, we focus on identifying a rich fragment of 
${\sf AC}^0$ that can in fact be embedded into the class of UHATs.
To achieve this, we use the characterization of 
${\sf AC}^0$ 
as the class of languages expressible in ${\rm FO}({\sf All})$, the extension of first-order logic (FO) with all numerical predicates defined in relation to the linear order of a word \cite{DBLP:books/daglib/0095988}.
 We show that UHATs recognize all languages
definable in ${\rm FO}({\sf Mon})$, the restriction of ${\rm FO}({\sf All})$ with {\em unary} numerical predicates only \cite{DBLP:journals/jcss/BarringtonILST05}.
The logic ${\rm FO}({\sf Mon})$ is highly expressive. Unlike FO, it can express non-regular languages like $\{a^n b^n \mid n > 0\}$. Remarkably, it contains all \emph{regular languages} within ${\sf AC}^0$, which includes examples like $(aa)^*$ --- a language not definable in FO. Additionally, our result subsumes the result of \cite{DBLP:conf/acl/YaoPPN20}, where it is shown that \emph{Dyck languages} of bounded nested depth can be recognized by UHATs. It is not hard to see that these languages are regular and belong to ${\sf AC}^0$, hence they are expressible in ${\rm FO}({\sf Mon})$. Our result also implies that UHAT is expressively more powerful than regular languages modulo letter-permutation (a.k.a. {\em Parikh images} \cite{Parikh,kozen-book}).

To establish the result that UHATs recognize all languages
definable in ${\rm FO}({\sf Mon})$, we take a slightly circuitous route: rather than directly formulating ${\rm FO}({\sf Mon})$ sentences as UHATs, we show that each formula in
${\rm LTL}({\sf Mon})$, the extension of 
{\em linear temporal logic} (${\rm LTL}$) 
\cite{DBLP:books/daglib/0007403-2} 
with arbitrary unary 
numerical predicates, 
can be equivalently represented as an UHAT. 
The proof for ${\rm FO}({\sf Mon})$
 then derives from Kamp's seminal theorem \cite{kamp1968tenselogic}, which establishes the equivalence between languages definable in ${\rm FO}$ and ${\rm LTL}$. The advantage of dealing with ${\rm LTL}$, in contrast to ${\rm FO}$, lies in the fact that all ${\rm LTL}$  formulas are unary in nature, i.e., they are interpreted as sets of positions on a word, unlike ${\rm FO}$ formulas which possess arbitrary arity. This property aligns well with the expressive capabilities of UHATs, facilitating a proof through structural induction.

While the fact that UHAT is in $\ACzero$ implies limited counting abilities of such encoders, recent work has shown that a slight extension of the hard attention mechanism can help in recognizing languages outside ${\sf AC}^0$ \cite{DBLP:journals/tacl/HaoAF22}. Instead of using unique hard attention, this model uses {\em average hard attention} (AHAT), which refers to the idea that the attention
 mechanism returns the uniform average value among all positions that maximize the attention. \emph{To what extent does AHAT enrich the counting ability of UHAT?} In answering this question, we introduce a logic named $\LTLCP$, which is an extension of LTL$({\sf Mon})$ that naturally incorporates  counting features. We show that any language that can be defined within $\LTLCP$ can also be identified by an AHAT. 
The logic $\LTLCP$ can express interesting languages lying outside ${\sf AC}^0$ including majority and parity (as far as we know, it have been shown before that parity can be accepted by an AHAT). More generally, our result implies that AHATs are equipped with a powerful counting ability: all permutation-closed languages over a binary alphabet and all permutation closures of regular languages  (which are in general not context-free) can be recognized by AHATs.

\paragraph{Related work.} There has been very little research on identifying logical languages that can be accepted by transformers. The only example we are aware of is the recent work by \cite{DBLP:conf/icml/0001CP23}, in which a variant of first-order logic with counting quantifiers is demonstrated to be embeddable into transformer encoders with a {\em soft attention} mechanism. The primary distinction between their work and our results is the choice of the attention mechanism. Additionally, the logic examined in their paper does not have access to the underlying word order being considered. This implies that some simple languages, such as $a^* b^*$, which are definable in ${\rm FO}$, are not definable in their logic. 

\paragraph{Proviso.} 
Some of the proofs in the paper are rather technical and lengthy. For such a reason we have relegated them to the appendix. 
 
\section{Background notions and results}
\label{sec:prelim}

\subsection{Transformer encoders}

We utilize a streamlined version of transformers, simplifying the model by abstracting certain features employed in more real-world scenarios. 

An {\em encoder layer} is a function that takes a sequence of vectors, $\mathbf{v}_0, \ldots, \mathbf{v}_{n-1}$, in $\mathbb{R}^d$ as input, where $d \geq 0$. It produces an output sequence of vectors, $\mathbf{v}'_0, \ldots, \mathbf{v}'_{n-1}$, in $\mathbb{R}^e$, with $e \geq 0$. We consider two types of encoder layers: {\em standard} and {\em ReLU}. Standard encoder layers resemble those found in most formalizations of transformer encoders. For the first part of the paper we assume that they employ a 
{\em unique hard} attention mechanism, meaning that a position only attends to the element with the highest attention score (breaking ties arbitrarily). On the other hand, ReLU encoder layers simply apply a ReLU function to the $k$th coordinate of each 
vector $\vb_i$. 
ReLU layers serve as a practical method for encoding logical formulas into transformers. A {\em transformer encoder} is then a concatenation of encoder layers. We define all these notions below.  

\paragraph{Standard encoder layer with unique hard attention.} A standard encoder layer is defined by three affine transformations, $A, B\colon\mathbb{R}^d\to\mathbb{R}^{d}$ and $C\colon \mathbb{R}^{2d}\to\mathbb{R}^{e}$. 
For $i\in\{0, \ldots, n-1\}$, we set
    \[\ab_i \gets \vb_{j_i},\]
where $j_i\in\{0, \ldots, n-1\}$ is the minimum element that maximizes the {\em attention score} 
$\langle A \vb_i, B \vb_j\rangle$ over $j\in\{0, \ldots, n-1\}$. 
The $a_i$s are often known as {\em attention vectors}. 
After that, we set
\[\vb'_i \gets C(\vb_i, \ab_i), \qquad  i = 0, \ldots, n-1.\]




It is useful to note that standard layers can do arbitrary position-wise affine transformations.

\paragraph{ReLU encoder layer.} A ReLU layer is given by $k\in\{1, 2, \ldots, d\}$. It just applies the ReLU 
function to the $k$th coordinate of each vector $\vb_i$. That is, assuming that $\vb_i = (v_i^1,\dots,v_i^d)$, then $\vb'_i \gets (v_i^1,\dots,v_i^{k-1},\max\{0,v_i^k\},v_i^{k+1},\dots,v_i^d)$, for $i = 0, \ldots, n-1$.   
The ReLU function can express the max of two numbers: 
$\max(x, y) = \max(0, x - y) + y$. 
This shows that with a constant number of ReLU layers, we can implement position-wise any function 
which is a composition of affine transformations and max.

\paragraph{Transformer encoder.} A {\em unique hard attention transformer encoder} (UHAT)\footnote{Some of the previous papers, for instance~\cite{DBLP:journals/tacl/HaoAF22}, allow to use in UHAT only rational numbers. We find this too restrictive because functions such as $\cos$ and $\sin$ are widely used in practice. Nevertheless, we stress that our results hold with this restriction, by taking good-enough approximations by rational numbers.} is defined simply as the repeated application of standard encoder layers with unique hard attention and ReLU encoder layers (with independent parameters). 


\subsection{Languages accepted by transformer encoders}

Next, we define how a transformer can be used to accept languages over a finite alphabet. 
This requires extending transformer encoders 
with three features: a function for representing alphabet symbols as vectors (which, for the purposes of this paper, we represent as one-hot encodings), another function that provides information about the absolute positions of these symbols within the input word, and a vector that is used for checking whether the word should be accepted or not. The function that provides information about positions is often referred to as a {\em positional encoding}, and it is essential for recognizing properties of ordered sequences of vectors. In fact, without positional encoding,  encoders treat input sequences as invariant to permutations \cite{DBLP:journals/jmlr/PerezBM21}. 

Consider a finite alphabet $\Sigma$ and let $T$ be an UHAT that takes a sequence of vectors over $\mathbb{R}^{d}$ as input and converts it into a sequence of vectors over $\mathbb{R}^e$.
A language $L \subseteq \Sigma^+$ is {\em accepted} by $T$, if there is an embedding function $f\colon \Sigma \to \mathbb{R}^d$, a positional encoding function $p\colon \mathbb{N}\times\mathbb{N} \to \mathbb{R}^d$, and 
a vector $\mathbf{t} \in \mathbb{R}^e$, such that 
for every $\bar w\in L$ we have $T(\bar w) > 0$, and for every $w\in \Sigma^+\setminus L$ we have $T(\bar w) < 0$. 
Here, $T : \Sigma^+ \to \mathbb{R}$ is defined as follows. Let $\bar w = a_0\ldots a_{n-1}\in\Sigma^n$, and
suppose the output of $T$ when given the input sequence
$f(a_0) + p(0, n), \, \ldots \, , f(a_{n-1}) + p(n-1, n)$ 
is the sequence $\vb_0, \dots, \vb_{n-1}$. Then we set $T(\bar w) = \langle \mathbf{t}, \vb_0 \rangle$.




\subsection{First order logic on words}

We assume familiarity with first-order logic (FO). Let $\Sigma$ be a finite alphabet. A word $\bar w = a_0 \cdots a_{n-1}$ in $\Sigma^+$ is represented as a structure $S_{\bar w}$ whose domain is $\{0,\dots,n-1\}$. This structure includes a binary relation $<$ that is interpreted as the linear order on the domain, and for each symbol $a \in \Sigma$, there is a unary relation $P_a$ containing positions $i = 0,\dots,n-1$ where $a_i = a$.
Given an ${\rm FO}$ {\em sentence} over words, that is, an ${\rm FO}$ formula without free variables, we denote the language of all words $\bar w \in \Sigma^+$ satisfying $S_{\bar w} \models \phi$ as $L(\phi)$. If an $L \subseteq \Sigma^+$ satisfies $L = L(\phi)$, for some ${\rm FO}$ sentence $\phi$, then we say that $L$ is {\em definable in ${\rm FO}$}. 

\begin{example} {\em 
First-order logic (FO) enables us to define certain languages of interest. Here, we present an illustrative example. Initially, we recognize that we can employ ${\rm FO}$ to define a relation ${\sf first}(x) := \neg \exists y (y < x)$ that exclusively holds true at the first position of a word. Correspondingly, we can define a relation ${\sf last}(x) := \neg \exists y (x < y)$ that holds solely at the last position of the word. Moreover, it is possible to define a binary relation ${\sf succ}(x,y) := x < y \wedge \neg \exists z (x < z \wedge z < y)$, which defines 
the successor relation within the domain. With these expressions, we can show that ${\rm FO}$ is capable of defining the language $(ab)^{+}$:
$$\exists x \, \big({\sf first}(x) \wedge P_a(x)\big) \ \wedge \ \exists x \,  
\big({\sf last}(x) \wedge P_b(x)\big) \ \wedge \ \forall x \forall y \, \big({\sf succ}(x,y) \rightarrow (P_a(x) \leftrightarrow P_b(y))\big).$$
That is, the first symbol of the word is an $a$, the last one is a $b$, every $a$ is followed by a $b$, and every $b$ is preceded by an $a$. \qed }
\end{example} 

\subsection{Unary numerical predicates}

It is known that ${\rm FO}$ sentences can only define regular languages. In turn, there are regular languages that are not definable in FO. An example is the language $(aa)^*$, which contains those words formed solely by the symbol $a$ that are of even length. However, there is a straightforward extension of ${\rm FO}$ that can define this language: all we need to do is add unary predicate ${\sf even}(x)$, which holds true at position $i$ in a word if and only if $i$ is even. 
In fact, extending ${\rm FO}$ with the predicate ${\sf even}(x)$ allows us to define the language $(aa)^*$ using the following formula, which indicates that the last symbol in the word satisfies the unary predicate ${\sf even}$: $\forall x P_a(x) \, \wedge \, \forall y ({\sf last}(y) \rightarrow {\sf even}(y))$. 

The extension of ${\rm FO}$ with unary numerical predicates can then be useful for defining languages. 
We define a {\em unary numerical predicate} $\Theta$ as an infinite family of functions
$$\theta_n : \{0,\dots,n\} \to \{0,1\}, \quad \quad n > 0.$$
Given a word $\bar w$ in $\Sigma^+$ of length $n$, for $n > 0$, we have that 
the predicate $\Theta(x)$ holds in position $i$ in $\bar w$
if and only if $\theta_n(i) = 1$ (so far, we do not use the value of $\theta_n$ at $n$ as positions are numbered from $0$ to $n-1$. We will use this value in Section 4). Notice that under our definition, the truth of a unary numerical predicate at position $i$ in the word $\bar w$ depends not only on $i$ but also on the length of the word $\bar w$. As we will explore further, this characteristic is advantageous for defining interesting languages in ${\rm FO}$ extended with arbitrary unary numerical predicates. 
Following the literature, we write ${\rm FO}({\sf Mon})$ for such an extension \cite{DBLP:journals/jcss/BarringtonILST05}. 

\begin{example} {\em 
Consider, for example, the non-regular language $\{a^nb^n \mid n > 0\}$. We show that it can be expressed in
${\rm FO}({\sf Mon})$ with the help of a unary numerical predicate $\Theta(x)$ such that $\theta_n(i) = 1$ iff $n$ is even and $i = n/2-1$. In fact, it suffices to use the formula: 
$$\exists x \, \big(\Theta(x) \, \wedge \, P_a(x) \, \wedge \, \forall y (y < x \rightarrow P_a(y)) \, \wedge \, \forall y (x < y \rightarrow P_b(y)) \big).$$
This formula expresses that the middle point $i$ of $\bar w$ exists, is labeled as $a$, and all positions smaller than $i$ are also labeled $a$, while all positions larger than $i$ are labeled as $b$. 
This example illustrates the significance of unary numerical predicates depending on both the position and the length of the word over which the formula is evaluated. \qed }
\end{example} 

The definition of the language $L(\phi) \subseteq \Sigma^+$ defined by an ${\rm FO}({\sf Mon})$ sentence $\phi$ is analogous to the one we provided for ${\rm FO}$.

\section{${\sf AC}^0$ languages accepted by UHATs}

\subsection{Not all languages in ${\sf AC}^0$ are accepted by UHATs.}

\cite{DBLP:journals/tacl/HaoAF22} proved that languages 
accepted by UHATs belong to  the circuit complexity class 
${\sf AC}^0$ , i.e., the class of languages accepted by families of Boolean circuits 
of unbounded fan-in, constant depth, and polynomial size. 
We combine results by \cite{DBLP:journals/apal/Ajtai83} and  \cite{DBLP:journals/tacl/Hahn20} to show that the opposite is not the case, i.e., there are ${\sf AC}^0$ languages that are not accepted by UHATs. 

As shown in~\cite{DBLP:journals/apal/Ajtai83}, there is an ${\sf AC}^0$-family of circuits $\{C_n\colon\{0,1\}^n\to\{0, 1\}\}_{n\in\mathbb{N}}$ such that for all $n$, the circuit $C_n$ accepts all strings with at at least $2n/3$ ones and rejects all strings with at most $n/3$. Consider a language {\em approximate majority}, consisting of strings accepted by circuits from $\{C_n\}$. This language is in ${\sf AC}^0$ by construction. However, as we state next, 
it cannot be recognized by an UHAT. This result is proved by using a property of UHATs established in \cite{DBLP:journals/tacl/Hahn20}.


\begin{prop} \label{notac0}
     There is no UHAT that accepts the language {\em approximate majority}. 
\end{prop}

\cite{viola2009approximate} shows that $\{C_n\}$ can be made polynomial-time computable, which implies the existence of a \emph{polynomial-time computable} language from ${\sf AC}^0$ that cannot be accepted by an UHAT.


\subsection{Main result: ${\rm FO}({\sf Mon})$ languages are accepted by UHATs}

Proposition \ref{notac0} tells us that not all ${\sf AC}^0$ languages are accepted by UHATs. In this section, we identify a significant subset of ${\sf AC}^0$ languages that can be accepted by UHATs. To accomplish this, we rely on the characterization of the class ${\sf AC}^0$ as those languages that can be defined in ${\rm FO}$ extended with arbitrary numerical predicates. Our main result establishes that as long as we restrict ourselves to unary numerical predicates, translation into UHATs is possible.


\begin{theorem} \label{main:fo}
Let $\Sigma$ be a finite alphabet and $\phi$ an ${\rm FO}({\sf Mon})$ sentence over words from the alphabet $\Sigma$. There is an UHAT that accepts $L(\phi)$. 
\end{theorem}

Proving this result by induction on ${\rm FO}({\sf Mon})$ formulas, which would be the most natural approach to tackle the problem, turns out to be difficult. The challenge arises because the ${\rm FO}({\sf Mon})$ formulas obtained by induction can have arbitrary arity, and transformer encoders do not seem capable of handling the requirements imposed by such formulas. To address this issue, we take a different approach. We employ Kamp's Theorem, which establishes that the languages definable in ${\rm FO}$ are precisely those that are definable in {\em linear temporal logic} (${\rm LTL}$) \cite{kamp1968tenselogic}. 

\subsection{Using ${\rm LTL}({\sf Mon})$ to prove our main result}

We first explain how ${\rm LTL}$ is defined, as this is crucial to understanding the remainder of the paper.
Let $\Sigma$ be a finite alphabet. ${\rm LTL}$ formulas over $\Sigma$ are defined as follows: if $a \in \Sigma$, then $a$ is an ${\rm LTL}$ formula. Additionally, ${\rm LTL}$ formulas are closed under Boolean combinations. Finally, if $\phi$ and $\psi$ are ${\rm LTL}$ formulas, then $\X \phi$ and $\phi \U \psi$ are also ${\rm LTL}$ formulas. Here, $\X$ is referred to as the {\em next} operator, and $\U$ as the {\em until} operator.

${\rm LTL}$ formulas are unary, i.e., they are evaluated over positions within a word. Let $\bar w = a_0 \cdots a_{n-1}$ be a word in $\Sigma^+$, and let $i = 0, \ldots, n-1$. We define the satisfaction of an ${\rm LTL}$ formula $\phi$ over $\bar w$ at position $i$, written as $(\bar w, i) \models \phi$, inductively as follows (omitting Boolean combinations):
\begin{itemize}
\item 
$(\bar w, i) \models a$ if and only if $a = a_i$, for $a\in \Sigma$.
\item 
$(\bar w, i) \models \X \phi$ if and only if $i < n-1$ and $(\bar w, i+1) \models \phi$. In other words, $\phi$ holds in the next position after $i$ (if such a position exists).
\item 
$(\bar w, i) \models \phi \U \psi$ if and only if there exists a position $j = i, \ldots, n-1$ for which $(\bar w, j) \models \psi$ and such that $(\bar w, k) \models \phi$ for every $k$ with $i \leq k < j$. That is, $\phi$ holds starting from position $i$ until the first position where $\psi$ holds (and a position where $\psi$ holds must exist). 
\end{itemize} 

We can extend ${\rm LTL}$ with unary numerical predicates in the same way we did it for ${\rm FO}$. Formally, we define ${\rm LTL}({\sf Mon})$ as the extension of ${\rm LTL}$ with every formula of the form $\Theta$, for $\Theta$ a unary numerical predicate. We write $(\bar w,i) \models \Theta$ to denote that $\theta_n(i) = 1$, where $n$ is the length of $\bar w$. 
If $\phi$ is an ${\rm LTL}({\sf Mon})$ formula over $\Sigma$, we write $L(\phi)$ for the set of  words $\bar w \in \Sigma^+$ with $(\bar w,0) \models \phi$.  

Kamp's Theorem establishes that for every ${\rm FO}$ sentence $\phi$ there exists an ${\rm LTL}$ formula $\psi$ such that $L(\phi) = L(\psi)$, and vice-versa. It is straightforward to see that this property extends to the logics ${\rm FO}({\sf Mon})$ and ${\rm LTL}({\sf Mon})$. 

\begin{prop} \label{kamp} \cite{kamp1968tenselogic}
For every ${\rm FO}({\sf Mon})$ sentence $\phi$ there exists an ${\rm LTL}({\sf Mon})$ formula $\psi$ such that $L(\phi) = L(\psi)$, and vice-versa. 
\end{prop}

Our proof of Theorem \ref{main:fo} is then derived directly from Proposition \ref{kamp} and the following result.

\begin{prop} \label{main:ltl}
Let $\Sigma$ be a finite alphabet and $\phi$ an ${\rm LTL}({\sf Mon})$ formula defined over words from the alphabet $\Sigma$. There is an UHAT $T$ that accepts $L(\phi)$. 
\end{prop}

Before proving this result, we make the following important remark regarding the positional encoding $p$ used by $T$ to accept 
$L(\phi)$. On a pair $(i,n) \in \mathbb{N} \times \mathbb{N}$ with $i < n$, we have that $p(i,n)$ is composed of elements $i$,
$\nicefrac{1}{(i+1)}$, $(-1)^i$, $\cos\left(\nicefrac{\pi(1 - 2^{-i})}{10}\right)$,  $\sin\left(\nicefrac{\pi(1 - 2^{-i})}{10}\right)$, and 
$\theta_n(i)$, for every unary numerical predicate $\Theta$
mentioned in $\phi$.

\begin{proof}[Proof of Proposition \ref{main:ltl}]
Let $\phi$ be a formula of $\LTLMon$. 
We say that a UHAT \emph{realizes $\phi$ position-wise} if, given a word $\bar w = a_0\ldots a_{n-1}\in\Sigma^+$,  the UHAT outputs a sequence:
\[\mathbb{I}\{(\bar w, 0)\models \phi\},\,\, \mathbb{I}\{(\bar w, 1)\models \phi\},\ \ldots \ ,\,\, \mathbb{I}\{(\bar w, n-1)\models \phi\};\] 
that is, a binary word indicating for which positions $\phi$ is true on $\bar w$ and for which is false. We show by structural induction that every $\LTLMon$  formula is realizable position-wise by some UHAT. 

Let us consider first the base cases. If $\phi = a$, for some $a\in \Sigma$, our goal is to obtain a sequence:
\[\mathbb{I}\{a_0 = a\},\,\,\mathbb{I}\{a_1 = a\},\ \ldots \ ,\,\,\mathbb{I}\{a_{n-1} = a\}.\]
This can easily be achieved by using a one-hot encoding as the embedding function. In turn, if $\phi = \Theta$, for $\Theta$ a unary numerical predicate, then $\phi$ can be realized position-wise using the corresponding positional encoding $p(i, n) = \theta_n(i)$.

We continue with Boolean combinations.
They can be implemented with a composition of ReLU layers and point-wise affine transformation:
$\lnot x =  1 - x$ and $x\lor y =  \frac{\max\{2x - 1, 2y - 1\} +  1}{2}$. 

For the cases when our formula is of the form $\X \phi$ or $\phi\U\psi$, we need the following lemma. 

    

\begin{lemma}
\label{last_pos}
    There is an UHAT
    that transforms each $x_0, \ldots, x_{n-1}\in\{0, 1\}$ as follows:
    \[x_0, \ldots, x_{n - 2}, x_{n-1} \mapsto x_0, \ldots, x_{n - 2}, 0.\]
\end{lemma}

Let us assume now that our formula is of the form $\X \phi$. It is enough to design a unique hard attention layer in which attention is always maximized at the next position. More precisely, we  construct an UHAT that outputs a sequence of vectors $\vb_1, \ldots, \vb_n\in\mathbb{R}^3$, and a linear transformation $A\colon\mathbb{R}^3\to\mathbb{R}^3$, such that $\arg\max_{j\in\mathbb{N}}\langle A \vb_i, \vb_j\rangle = \{i + 1\}$, 
    for $i = 0, \ldots, n - 2$. This will allow us to ``send'' $\mathbb{I}\{(\bar w,i+1)\models \phi\} =  \mathbb{I}\{(\bar w,i)\models \X\phi\}$  
    to the $i$th position, for $i = 0, \ldots, n - 2$. It only remains then to apply Lemma \ref{last_pos} to obtain $0 = \mathbb{I}\{(\bar w,n-1)\models \X\phi\}$ at the last  position.

    Using our positional encoding  and an affine position-wise transformation, we can obtain:
    \[\vb_i = \Big(\cos\left(\frac{\pi(1 - 2^{-i})}{10}\right), \,\, \sin\left(\frac{\pi(1 - 2^{-i})}{10}\right), \,\,(-1)^i\cdot 10 \Big).\]
Let $A$ be a linear transformation that reverses the third coordinate. Observe that:
\[\langle A\vb_i, \vb_j\rangle =\cos\left(\frac{\pi(2^{-i}-2^{-j})}{10}\right) + (-1)^{i + j + 1}\cdot 10.\]

We claim that, for a fixed $i$, this quantity is maximized at $j = i + 1$. First, those $j$s that have the same parity as $i$ (in particular, $j = i$) cannot achieve the maximum because the second term is $-10$. For $j$s with a different parity, we have $\langle A \vb_i, \vb_j\rangle = \cos\left(\nicefrac{\pi(2^{-i}-2^{-j})}{10}\right) + 10$. Since all angles are in $[-\pi/10, \pi/10]$, this quantity is maximized when $|2^{-i} - 2^{-j}|$ is minimized. For $j < i$, the last quantity is at least $2^{-i}$, and for $j > i$, the minimum of this quantity is $2^{-i -1}$, achieved at $j  = i + 1$.

Let us finally assume that our formula is of the form $\phi \U \psi$. 
For a given $i = 0, \ldots, n-1$, let $j_i$ be the minimal $j\in\{i, \ldots, n-1\}$ such that $(\bar w,j)\not\models\phi$, and if no such $j$ exists, $j_i = n-1$. Observe that $(\bar w, i)\models$ $\phi\U\psi$ if and only if $(\bar w, j_i)\models \psi$. To show the lemma, it is enough to create a unique hard attention layer, where for every position $i$ the attention is maximized at $j_i$.

         Due to the Lemma \ref{last_pos}, we may assume, without loss of generality, that $(\bar w, n-1)\not\models \phi$.
         Then for every $i$, there exists at least one $j\in\{i, \ldots, n-1\}$ such that $(\bar w,j)\not\models\phi$, and then $j_i$ can be defined as the minimal such $j$, without any clauses.

Using our positional encoding and the induction hypothesis, we can obtain a sequence of vectors $\vb_1, \ldots, \vb_n\in\mathbb{R}^4$ such that:
\[\vb_i = 
\Big(  \cos\left(\frac{\pi(1 - 2^{-i})}{10}\right),\,\,\sin\left(\frac{\pi(1 - 2^{-i})}{10}\right),\,\, 1,\,\, \mathbb{I}\{w,i\models\phi\}\Big).
\]
Consider a linear transformation $B\colon\mathbb{R}^4\to\mathbb{R}^4$ such that 
\[B \vb_i = \Big(  \cos\left(\frac{\pi(1 - 2^{-i})}{10}\right),\,\,\sin\left(\frac{\pi(1 - 2^{-i})}{10}\right),\,\, -10 \cdot\mathbb{I}\{w,i\models\phi\},\,\,0\Big).\]
Observe that 
\[
\langle \vb_i, B\vb_j\rangle = \cos\left(\frac{\pi(2^{-i} - 2^{-j})}{10}\right) - 10\cdot \mathbb{I}\{\bar w,j\models\phi\}.\]
We claim that this expression is maximized at $j = j_i$. First, because of the last term in it, it cannot be maximized at $j$ with $(\bar w, j)\models \phi$. It remains to show that among the $j$s with $(\bar w, j)\not\models \phi$, this quantity is minimized on the minimal $j$ which is at least $i$. In fact, in this case we have $\langle \vb_i, B\vb_j\rangle = \cos\left(\frac{\pi(2^{-i} - 2^{-j})}{10}\right)$. 
All the angles in question are in $[-\pi/10, \pi/10]$, so the cosine is maximized when $|2^{-i} - 2^{-j}|$ is minimized. Now, this absolute value is at least $2^{-i}$ when $j < i$. In turn, this absolute value is smaller than $2^{-i}$ for $j\ge i$, and it is the smaller the smaller is $j$, as required.
\end{proof} 

\subsection{Applications of our main result}

We show two applications of our main result. First, UHATs accept all regular languages in ${\sf AC}^0$. Second, UHATs are strictly more expressive than regular and context-free languages in terms of the acceptance of languages up to letter-permutation. 

\paragraph{Regular languages in ${\sf AC}^0$.}
There is an important fragment of ${\rm FO}({\sf Mon})$ which is interesting in its own right. 
This is the logic ${\rm FO}({\sf Mod})$, i.e., 
the extension of ${\rm FO}$ with unary numerical predicates of the form ${\sf Mod}_p^r$, for $p > 1$ and $0 \leq r \leq p-1$. We have that ${\sf Mod}_p^r(i) = 1$ if and only if $i \equiv r \, ({\rm mod} \, p)$.  In fact, by using a characterization given in \cite{DBLP:journals/jcss/BarringtonCST92}, one can show that the languages definable in  ${\rm FO}({\sf Mod})$ are precisely the regular languages within ${\sf AC}^0$. Then:  

\begin{corollary} 
Let $L \subseteq \Sigma^+$ be a regular language in 
${\sf AC}^0$. There is an UHAT that accepts $L$. 
\end{corollary}

\input{parikh_prelim}

\section{Languages beyond ${\sf AC}^0$}

Transformer encoders with unique hard attention can only recognize languages in ${\sf AC}^0$, but a slight extension of the attention mechanism allows to recognize languages lying outside such a class \cite{DBLP:journals/tacl/HaoAF22}. In this section, we show that in fact such an extended model can recognize all languages definable in a powerful logic that extends ${\rm LTL}$ with counting features. This logic can express interesting languages outside ${\sf AC}^0$, such as {\em majority} and {\em parity}.  

\subsection{Average hard attention}

For the results in this section, we consider an extended version of transformer encoders that utilize an {\em average hard attention mechanism} \cite{DBLP:journals/jmlr/PerezBM21,DBLP:journals/tacl/HaoAF22}. Following the literature, we call these AHAT. The difference between UHAT and AHAT only lies at the level of the standard encoder layers, which are now defined as follows. 

\paragraph{Standard encoder layer with average hard attention.} As before, these layers are defined by three affine transformations, $A, B\colon\mathbb{R}^d\to\mathbb{R}^{d}$ and $C\colon \mathbb{R}^{2d}\to\mathbb{R}^{e}$. 
For every $i\in\{0, \ldots, n-1\}$, 
we define $S_i$ as the set of positions $j \in \{0,\dots,n-1\}$ 
that maximize  
$\langle A \vb_i, B \vb_j\rangle$. 
We then set
    \[\ab_i \, \gets \, \Big(\sum_{j \in S_i} \vb_{j}\Big)/|S_i|.\]
After that, we set $\vb'_i \gets C(\vb_i, \ab_i)$, for each  
$i = 0, \ldots, n-1$. 
That is, attention scores under 
average hard attention return the uniform average value among all positions that maximize attention.   

We also use \emph{future positional masking} that allows us to take into account only positions up to $i$. If the future positional masking is used, the sets $S_i$ are defined as sets of positions $j\in\{0, 1, \ldots, i\}$ that maximize $\langle A\vb_i, B\vb_j\rangle$.  Positional masks have been employed on several occasions in theoretical papers~\cite{DBLP:conf/acl/YaoPPN20,DBLP:conf/emnlp/BhattamishraAG20,DBLP:journals/tacl/HaoAF22} as well as in practice, for example, for training GPT-2~\cite{radford2019language}.



\subsection{${\rm LTL}$ extended with counting terms}

We present here $\LTLCP$, an extension of ${\rm LTL}({\sf Mon})$ that allows us to define counting properties over words in a simple manner. This requires the introduction of {\em counting terms} as defined next. 

\paragraph{Counting terms.} 
Suppose $\phi$ is a unary formula. Then 
$\overleftarrow{\#\phi}$ and $\overrightarrow{\#\phi}$ are counting terms. The interpretation of these terms in position $i$ of  a word $\bar w$ of length $n$ is defined as follows: 
\begin{align*}
    \overleftarrow{\#\phi}(\bar w, i) & \ = \ \left|\{j \in\{0, \ldots, i \}\mid (\bar w, j) \models \phi \}\right|,\\ \overrightarrow{\#\phi}(\bar w, i) &\ = \ \left|\{j \in\{i, \ldots, n-1\}\mid (\bar w, j) \models \phi \}\right|.
    \end{align*}
    That is, $\overleftarrow{\#\phi}(\bar w, i)$ is the number of positions to the left of $i$ (including $i$) that satisfy $\phi$, while $\overrightarrow{\#\phi}(\bar w, i)$ is the number of positions to the right of $i$ (including $i$) that satisfy $\phi$. Notice that, for words of length $n$, counting terms take values in $\{0, 1, \ldots, n\}$. 

\paragraph{Counting formulas.}
With counting terms and unary numerical predicates we can create new formulas in the following way. Let $\phi$ be a unary formula and $\Theta$ a unary numerical predicate. We define new formulas $\Theta(\overleftarrow{\#\phi})$ and $\Theta(\overrightarrow{\#\phi})$. The interpretation of such formulas on position $i$ of a word $\bar w$ of length $n$ is as follows: 
$$(\bar w,i) \models \Theta(\overleftarrow{\#\phi}) \ \Leftrightarrow \ \theta_n(\overleftarrow{\#\phi}(\bar w,i)) = 1 \ \ \ \ \ \ \ \ \ \ (\bar w,i) \models \Theta(\overrightarrow{\#\phi}) \ \Leftrightarrow \ \theta_n(\overrightarrow{\#\phi}(\bar w,i)) = 1.$$
That is, the number of positions to the left (resp., right) of $i$ (including $i$) that satisfy $\phi$ satisfies the predicate $\Theta$. As counting terms can take value $n$, the value of $\theta_n$ on $n$ becomes useful. 

We also incorporate into our logic the possibility of checking linear inequalities with integer coefficients over counting terms. More specifically, for any finite set of unary 
formulas $\phi_1, \ldots, \phi_k, \psi_1, \ldots, \psi_k$, and for any coefficients $c_1, \ldots, c_k, d_1, \ldots, d_k\in\mathbb{Z}$ we can create a formula: $$\sum_{j = 1}^ k c_j \cdot \overleftarrow{\#\phi_j} 
\, + \, 
\sum_{j=1}^k d_j \cdot \overrightarrow{\#\psi_j} 
\ \ge \ 0,$$ which is interpreted as follows:
\begin{align*}
(\bar w, i)&\models \sum_{j = 1}^ k c_j \cdot \overleftarrow{\#\phi_j} 
\, + \, 
\sum_{j=1}^k d_j \cdot \overrightarrow{\#\psi_j} 
\, \ge \, 0 
\,   
 \iff & \\  
& 
\sum_{j=1}^k 
c_j\cdot\overleftarrow{\#\phi_j}(\bar w, i) 
\, + \, \sum_{j=1}^k d_j \cdot\overrightarrow{\#\psi_j}(\bar w, i) 
\, \ge \, 0.
\end{align*}

\paragraph{The logic $\LTLCP$.}
We denote by $\LTLCP$ the logic that is recursively defined as follows: 
\begin{itemize} 
\item 
Every formula ${\rm LTL}({\sf Mon})$ is also an $\LTLCP$ formula. 
\item 
Boolean combinations of $\LTLCP$ formulas are $\LTLCP$ formulas.
\item 
If $\phi$ and $\psi$ are $\LTLCP$ formulas, then so are $\X \phi$ and $\phi\U\psi$. 
\item If $\phi$ is an $\LTLCP$ formula and $\Theta$ is a unary numerical predicate, then $\Theta(\overleftarrow{\#\phi})$ and $\Theta(\overrightarrow{\#\phi})$ are $\LTLCP$ formulas. 
\item If $\phi_1, \ldots, \phi_k, \psi_1, \ldots, \psi_k$ are formulas of $\LTLCP$, then $\sum_{j = 1}^ k c_j \cdot \overleftarrow{\#\phi_j} 
\, + \, 
\sum_{j=1}^k d_j \cdot \overrightarrow{\#\psi_j} 
\ \ge \ 0,$  is a formula of $\LTLCP$.
\end{itemize}

\subsection{${\rm LTL}({\bf C})$ definable languages are accepted by encoders} 

Next, we state the main result of this section: languages definable by $\LTLCP$ formulas are accepted by transformer encoders with average hard attention. 
                                       
\begin{theorem}
\label{thm:terms}
Let $\Sigma$ be a finite alphabet and $\phi$ an $\LTLCP$ formula defined over words from the alphabet $\Sigma$. There is an AHAT $T$ that accepts $L(\phi)$. 
\end{theorem}

 As a corollary to Theorem \ref{thm:terms}, we show that AHATs are rather powerful in counting. To make this claim more formal, we study \emph{permutation-closed} languages, i.e., languages $L$ such that $\bar v \in L$ iff any letter-permutation of $\bar v$ is in $L$. 
 For a language $L$, we write $perm(L)$ to be the permutation-closure of $L$, i.e., $perm(L) = \{ \bar w : \Parikh(\bar w) = \Parikh(\bar v), \text{ for some $\bar v \in L$} \}$. Observe that $perm((abc)^*)$ consists of all strings with the same number of occurrences of $a$, $b$, and $c$; this is not even context-free. Owing to Parikh's Theorem, to recognize $perm(L)$, where $L$ is a regular language, an ability to perform letter-counting and linear arithmetic reasoning (i.e. semilinear set reasoning) is necessary. AHATs possess such an ability, as shown by the following corollary.
 
 \begin{corollary}
    The permutation closure $perm(L)$ of any regular language $L$ is accepted by an AHAT. Moreover, any permutation-closed language over a binary alphabet is accepted by an AHAT. 
    \label{cor:parikh_ahat}
 \end{corollary}

 Both {\em majority} and {\em parity} are permutation-closed and are over a binary alphabet. Hence, by the previous result, they are both accepted by AHATs. While for {\em majority} this was known \cite{DBLP:journals/tacl/HaoAF22}, the result for {\em parity} is new.


\section{Conclusions and future work}

We have conducted an 
investigation 
of the problem of which languages can be accepted by transformer encoders with hard attention. For UHATs, we have demonstrated that while they cannot accept all languages in ${\sf AC}^0$, they can still accept all languages in a 'monadic' version of it defined by the logic ${\rm FO}({\sf Mon})$. Crucial to the proof of this result is the equivalence between ${\rm FO}$ and ${\rm LTL}$, as provided by Kamp's Theorem. In turn, we have shown that AHATs are capable of expressing any language definable in a powerful counting logic, $\LTLCP$, that can express properties beyond ${\sf AC}^0$. This implies, among other things, that the {\em parity} language can be accepted by an AHAT. 

Several interesting problems remain open in our work, especially regarding characterizations of the classes we have studied. To begin, are there languages accepted by UHATs that cannot be defined in ${\rm FO}({\sf Mon})$? Additionally, does there exist a language in the circuit complexity class ${\sf TC}^0$, the extension of ${\sf AC}^0$ with majority gates, that cannot be recognized by AHATs? Lastly, is there a language that can be accepted by an AHAT but cannot be defined in $\LTLCP$?
   
\bibliographystyle{acm}
\bibliography{biblio}

\newpage

\section*{Appendix}

\begin{proof}[Proof of Proposition \ref{notac0}]
 As Hahn showed, for every $\varepsilon > 0$ and $L > 0$ there exists $c \geq 0$ such that, for all larger enough $n$, if we consider as inputs binary strings of length $n$, for every UHAT $T$ consisting of  $L$ layers, there exists a fixation of $\varepsilon n$ input bits such that, under this fixation, the output of $T$ is determined by $c$ unfixed bits ~\cite{DBLP:journals/tacl/Hahn20}. However, it cannot hold for an UHAT recognizing {\em approximate majority}, for example, when $\varepsilon = 1/10$. Regardless of how we fix $n/10 + c$ input bits, if we fix the remaining bits to $0$s, the circuit $C_n$ rejects our string, and if we fix them to $1$s, it accepts our string, even though the output of the UHAT remains unchanged.
 \end{proof} 

\begin{proof}[Proof of Lemma \ref{last_pos}]
    At position $i = 0, \ldots, n-1$, this transformation can be written as follows:
    \[x_i\mapsto x_i - \max\{0, x_i+ i - (n+1)\}.\]
    It can easily be done with ReLU layer, using a positional encoding $p(i, n) = i-n$. However, it can also be done with a positional encoding that does not depend on $n$, for example $p(i) = (i, 1/(i+1))$. We just have to ``transmit'' $n-1$ to every position in the UHAT. For that, it is enough to have a unique hard attention layer, where attention in every position is maximized at $j = n-1$ (which allows that to ``send'' $n$ to every position). For instance, consider $\vb_i = 1/(i+1)$, $A(x) = -x$, and observe that:
        \[\arg\max_{j = 0, \ldots,n-1} \langle A \vb_i, \vb_j\rangle =\arg\max_{j = 0, \ldots,n-1} -\frac{1}{(i+1)(j+1)} = \{n-1\}\]
        for every $i = 0, \ldots, n-1$.
    This finishes the proof of the lemma.
\end{proof}

\input{parikh_uhat_appendix}

\begin{proof}[Proof of Theorem \ref{thm:terms}]
    As before, we are proving that every formula $\phi$ of  $\LTLCP$ can be computed position-wise by some AHAT encoder, via structural induction. We have already shown how to do induction for all operators of ${\rm LTL}({\sf Mon})$. In our proof, attention was always maximized at the unique $j$, and in this case, there is no difference between unique and average hard attention.

    It remains to show the same for operators that are in $\LTLCP$ but not in ${\rm LTL}({\sf Mon})$. First, we show that given a formula $\phi$, computed position-wise by some AHAT, there is also an AHAT that computes $\overleftarrow{\#\phi}$ and $\overrightarrow{\#\phi}$ position-wise.

    Using future positional masking and equal weights, we can compute at position $i$ the quantity:
    \[ y_i = \frac{\phi(w, 0) + \ldots + \phi(w, i)}{i+1} = \frac{\overleftarrow{\#\phi}(w, i)}{i+1}, \qquad i = 0, 1,\ldots, n -1.\]
    Next, we have to compute 
    \[z_i = \frac{\big(\overleftarrow{\#\phi}(w, i) -\phi(w, i)\big)}{i+1}.\]
    This can be achieved as follows:
    \[z_i = y_i - \frac{\phi(w, i)}{i+1} = y_i - \min\left\{\phi(w, i), \frac{1}{i+1}\right\}.\]
    As our positional encoding includes $1/(i + 1)$, this computation is a composition of ReLU and affine transformations.
    
    Our next goal is to get rid of the coefficient $1/(i +1)$.
For that, we create a layer with the following attention function:
    \begin{equation}
    \label{get_rid}
    \langle A\vb_i, B\vb_j\rangle = 2j\cdot z_i - \frac{j^2}{i + 1}, \qquad i, j = 0, \ldots, n - 1.
    \end{equation}
    Such attention function is possible because \eqref{get_rid} is a bilinear form of $\vb_i$ and $\vb_j$. Indeed, $\vb_i$ contains $1/(i+1)$ and $\vb_j$ contains $j, j^2$ due to our positional encoding, and also $\vb_i$ contains $z_i = \frac{\big(\overleftarrow{\#\phi}(w, i) -\phi(w, i)\big)}{i+1}$.
    
Denoting $d_i = \overleftarrow{\#\phi}(w, i) -\phi(w, i)$, we get that \eqref{get_rid} is equal to
\[\langle A\vb_i, B\vb_j\rangle = 2j \cdot\frac{d_i}{i+1} - \frac{j^2}{i + 1} = \frac{-(d_i - j)^2 + d_i^2}{i +1}.\]

Observe that $d_i = \overleftarrow{\#\phi}(w, i) -\phi(w, i)$ takes values in $\{0, \ldots, n - 1\}$. Hence, for a fixed $i$, the quantity \eqref{get_rid} is uniquely maximized at $j = d_i$. In this way, we get $j = d_i$ to position $i$. Adding $\phi(w, i)$ to $d_i$, we get $\overleftarrow{\#\phi}(w, i)$.  To get $\overrightarrow{\#\phi}(w, i)$ to position $i$, we observe that:
     \begin{align*}
        \overrightarrow{\#\phi}(w, i) &= (\phi(w,0) + \ldots + \phi(w, n-1)) - ((\phi(w,0) + \ldots + \phi(w, i-1)) \\
        &=\overleftarrow{\#\phi}(w, n-1) - d_i.
    \end{align*}
    This is computable at position $i$ because $\overleftarrow{\#\phi}(w, n-1)$  can be ``sent'' to all positions via the attention function, always maximized at the last position (see the proof of Lemma \ref{last_pos}). 

Our next goal is: given a formula $\phi$, computable position-wise by some AHAT, and a unary numerical predicate $\Theta$, provide an AHAT that computes $\Theta(\overleftarrow{\#\phi})$ and $\Theta(\overrightarrow{\#\phi})$ position-wise. As we have already shown, we can assume that we already have counting terms $\overleftarrow{\#\phi}$ and $\overrightarrow{\#\phi}$ computed position-wise. Next, we create a layer with the following attention function:
 \[\langle A\vb_i, B\vb_j\rangle = 2j\cdot \overrightarrow{\#\phi}(w, i) - j^2 = -(j -   \overrightarrow{\#\phi}(w, i))^2 + \overrightarrow{\#\phi}(w, i)^2.
    \]
    Again, this is possible because this expression is a bilinear form of $\vb_i$ and $\vb_j$, due to our positional encoding. It is maximized at $j_i = \min\{n -1,  \overrightarrow{\#\phi}(w, i)\}$ (when the counting term is equal to $n$, since we do not have a position indexed by $n$, the maximizing position will be $j_i = n-1$). Having $\Theta$ included in the positional encoding, we can get $j_i$ and $\theta_n(j_i)$ to the $i$th position. 
    Observe that:
    \[\theta_n(\overrightarrow{\#\phi}(w, i)) = (\mathbb{I}\{\overrightarrow{\#\phi}(w, i) \le n - 1\}\land \theta_n(j_i))\lor (\lnot \mathbb{I}\{\overrightarrow{\#\phi}(w, i) \le n - 1\}\land \theta_n(n))\]
    Since in our positional encoding, $\theta_n(n)$ is included in every position, and since position-wise Boolean operations can be done by an AHAT, it remains to compute the indicator $\mathbb{I}\{\overrightarrow{\#\phi}(w, i) \le n - 1\}$. Transmitting $n$ once again to every position, we can write:
    \[\mathbb{I}\{\overrightarrow{\#\phi}(w, i) \le n - 1\} = \min\{1, n -\overrightarrow{\#\phi}(w, i)\}.\]
    
    This quantity can be computed by a composition of ReLU and affine transformations.
    We can get $\theta_n(\overleftarrow{\#\phi}(w, i))$  to the $i$th position analogously.

    Finally, we have to check that linear inequalities over counting terms can be done in AHAT. Given formulas $\phi_1, \ldots, \phi_k, \psi_1, \ldots, \psi_k$ already computed position-wise by some AHAT, we have to provide an AHAT that computes the formula $\sum_{j = 1}^ k c_j \cdot \overleftarrow{\#\phi_j} 
\, + \, 
\sum_{j=1}^k d_j \cdot \overrightarrow{\#\psi_j} 
\ \ge \ 0$ position-wise. After computing counting terms for $\phi_1,\ldots, \phi_k,\psi_1, \ldots, \psi_k$, we first can compute their linear combination, using affine position-wise transformations:
    \[l_i = \sum_{j = 1}^ k c_j \cdot \overleftarrow{\#\phi_j}(w, i)
\, + \, 
\sum_{j=1}^k d_j \cdot \overrightarrow{\#\psi_j}(w, i). \]
    Since coefficients are integral, $l_i$ is integral as well, so we get: \[\mathbb{I}\{l_i\ge 0\} = \max\{\min\{0, l_i\} + 1, 0\}.  \]
    The last expression can be computed via composition of ReLU and affine transformations.

 \end{proof}

 \input{parikh_ahat}

\end{document}

%% file: comm-anthony.tex
\newcommand{\ialphabet}{\Sigma}
\newcommand{\vecV}{\textbf{v}}
\newcommand{\N}{\mathbb{N}}
\newcommand{\defn}[1]{\emph{#1}}
\newcommand{\Lang}{L}
\newcommand{\Parikh}{\mathcal{P}}
\newcommand{\ACzero}{\ensuremath{{\sf AC}^0}}

\newcommand{\FOMon}{{\rm FO}({\sf Mon})}

\newcommand{\LTLMon}{{\rm LTL}({\sf Mon})}

\newcommand{\LTLCP}{{\rm LTL}({\bf C},\mathbf{+})}

%% file: abstract.tex
\begin{abstract}
We contribute to the study of formal languages that can be recognized by transformer encoders. We focus on two self-attention mechanisms: (1) UHAT (Unique Hard Attention Transformers) and (2) AHAT (Average Hard Attention Transformers). UHAT encoders are known to  recognize only languages inside the circuit complexity class ${\sf AC}^0$, i.e., accepted by a family of poly-sized and depth-bounded boolean circuits with unbounded fan-ins. On the other hand, AHAT encoders can recognize languages outside ${\sf AC}^0$), but their expressive power still lies within the bigger circuit complexity class ${\sf TC}^0$, i.e., ${\sf AC}^0$-circuits extended by majority gates.
We first show a negative result that there is an  ${\sf AC}^0$-language that cannot be recognized by an UHAT encoder. On the positive side, we show that UHAT encoders can recognize a rich fragment of ${\sf AC}^0$-languages, namely, all languages definable in first-order logic with arbitrary unary numerical predicates. This logic, includes, for example, all regular languages from  ${\sf AC}^0$. We then show that AHAT encoders can recognize all languages of our logic even when we enrich it with counting terms. We apply these results to derive new results on the expressive power of UHAT and AHAT up to permutation of letters (a.k.a. Parikh images).
\end{abstract}

%% file: parikh_prelim.tex
\paragraph{Recognizing regular languages up to letter-permutation.}
Although not all regular languages are accepted by UHATs (e.g. {\em parity}), we can 
use Theorem \ref{main:fo} to show that, up to letter-permutation, UHAT is in fact strictly more powerful than regular and context-free languages.

To formalize our result, we recall the notion of semilinear sets and the Parikh image of a language.
A \defn{linear set} $S$ is a subset of $\N^d$ (for some positive integer $d$, called \defn{dimension}) of the form 
\[
    \vecV_0 + \sum_{i=1}^r \vecV_i \N \ := \ \{ \vecV_0 + \sum_{i=1}^r k_i \vecV_i : k_1,\ldots,k_r \in \N \}
\]
for some vectors $\vecV_0,\ldots,\vecV_r \in \N^d$. A \defn{semilinear set} $S$ over $\N^d$ is a finite union of linear sets over $\N^d$. Semilinear sets have a very tight connection to formal languages through the notion of the \defn{Parikh image} a language $L$ \cite{Parikh}, which intuitively corresponds to the set of ``letter-counts'' of $L$.
More precisely, consider the alphabet $\ialphabet = \{ a_1,\ldots,a_d \}$ and a language $\Lang$ over $\ialphabet$. For a word $w \in \ialphabet$, let $|w|_{a_i}$ denotes the number of occurrences of $a_i$ in $w$. 
The \emph{Parikh image} $\Parikh(\Lang)$ of $\Lang$ is defined to be the set of tuples $\vecV = (|w|_{a_1},\ldots,|w|_{a_d}) \in \N^d$ for some word $w \in L$. For example, if $L = \{ a^n b^n : n \geq 0 \}$ and $L' = (ab)^*$, then
$\Parikh( L ) = \Parikh( L')$. In this case, we say that $L$ and $L'$ are \emph{Parikh-equivalent}. Note that $L'$ is regular, while $L$ is context-free but not regular. This is not a coincidence based on the celebrated Parikh's Theorem (cf. \cite{Parikh}, also see \cite{kozen-book}).

\begin{prop}[\cite{Parikh}]
The Parikh images of both regular and context-free languages coincide with semilinear sets.
\end{prop}

In other words, although context-free languages are strict superset of regular languages, they are in fact equally powerful up to letter-permutation.
What about UHATs? We have that they are strictly more powerful than regular and context-free languages up to letter-permutation.
\begin{prop}
    Each regular language has a Parikh-equivalent language accepted by an UHAT. In turn, there is an UHAT language with no Parikh-equivalent regular language.
    \label{prop:parikh_uhat}
\end{prop}


%% file: parikh_uhat_appendix.tex
\begin{proof}[Proof of Proposition \ref{prop:parikh_uhat}]
\textbf{Upper bound:} We first show that every regular language over $\ialphabet = \{a_1,\ldots,a_d\}$ has a Parikh-equivalent language in UHAT. By Parikh's Theorem, the Parikh image of this given regular language is represented by a semilinear set $S$ in dimension $d$.
Our proof employs Theorem \ref{main:fo}. Since $\FOMon$ is closed under disjunction, it suffices to consider only linear sets $S$. Thus, take an arbitrary linear set 
$S = \vecV_0 + \sum_{i=1}^r \vecV_i \N$, where $\vecV_i$ ($i > 0$) is a non-zero vector. We will give a language $L$ over the alphabet of $\ialphabet = \{a_1,\ldots,a_d\}$ definable in $\FOMon$ (thus UHAT-recognizable, by Theorem \ref{main:fo}) such that $\Parikh(L) = S$. We will use the linear set $S = (1,1,0) + (2,0,1)\N$ as a running example.

For $i = 0,\ldots,r$ and $j =1,\ldots,d$, define $v_i^j$ to be the natural number corresponding to the $j$th argument of $\vecV_i$. Define $w_i^j$ to be the string $a_j^{v_i^j}$, i.e., $a_j$ repeated $v_i^j$ times, while $\ell_i$ denotes the ``length abstraction'' of $\vecV_i$, i.e., $\ell_i := \sum_{j=1}^d v_i^j$. Finally, let $w_i$ be the concatenation of $w_i^1,\ldots,w_i^d$. Using our example of $S = (1,1,0) + (2,0,1)\N$, then we have $w_0 = a_1a_2$ and $w_1 = a_1a_1a_3$. We also have $\ell_0 = 2$ and $\ell_1 = 3$. 

Next we define the language $L$ as follows:
\[
L := w_0\cdot w_1^*\cdots w_r^*
\]
Using our running example, $L$ would be $a_1a_2(a_1a_1a_3)^*$.
It is easy to see that $\Parikh(L) = S$.

To show that this language is in $\FOMon$-definable, we demonstrate that it is regular and belongs to $\ACzero$. It is regular because it is defined through concatenation and Kleene star. Since $\ACzero$ is closed under concatenation\footnote{if we have $\ACzero$-circuits $C_1, C_2$ for languages $L_1, L_2$, we can construct an $\ACzero$-circuit $C$ for their concatenation as follows: $C(x_1\ldots x_n) = \bigvee_{i = 1,\ldots, n} (C(x_1\ldots x_i)\land C(x_{i+1}\ldots x_n)) $.} it remains to show that languages of the form $w^*$, where $w$ is a word, are in $\ACzero$. We only have to care about input lengths that are multiples of $|w|$, for other input lengths the language is empty. Then we split the input into blocks of $|w|$ letters. We just need an $\ACzero$-circuit, checking that every block coincides with $w$. For example, this can be done with an AND over blocks of constant-size circuits, checking equality to $w$.

\smallskip
\noindent
\textbf{Lower bound:} An example of a language that is in $\FOMon$ (and so in UHAT) whose Parikh image is not semilinear (and therefore, no Parikh-equivalent regular language) is 
\[
L = \{ a^k : \text{$k$ is a prime number} \}.
\]
Note that $\ialphabet = \{a\}$.
This can be easily defined in $\FOMon$ using the unary predicate $\Theta := \{ k \in \N: \text{$k+1$ is a prime number}\}$ as follows: $\exists x \Theta(x)\land \lnot \exists y > x$.
\end{proof}

%% file: parikh_ahat.tex
\begin{proof}[Proof of Corollary \ref{cor:parikh_ahat}]
We show that permutation-closed languages over binary alphabets and languages of the form $perm(L)$, where $L$ is a regular language, are expressible in $\LTLCP$.

First, assume that $L$ is a permutation-closed language over a binary alphabet $\{a, b\}$. Then whether or not a word $\bar w$ belongs to $L$ is determined by the length of $w$ and the number of $a$'s in $w$. In other words, there a numerical predicate $\Theta$ such that for every $n$ and for every $\bar w\in\{a, b\}^n$, we have $\bar w\in L$ if and only if $\theta_n(|\bar w|_a) = 1$ (recall that for a word $\bar w$ and for a letter $a$, the expression $|\bar w|_a$ denotes the number of occurrences of $a$ in $\bar w$). Thus, $L$ is expressible by the  formula $\Theta(\overrightarrow{\# a})$.

We now show that every language of the form $perm(L)$, where $L$ is regular, is expressible in $\LTLCP$.

As shown in \cite{Parikh}, if $L$ is a regular language over the alphabet $\ialphabet = \{a_1,\ldots,a_d\}$, then 
\[
perm(L) = \{ w : \Parikh(w) \in S \},
\]
for some semilinear set $S$ of dimension $d$. Semilinear sets correspond precisely to sets of tuples that are definable in Presburger Arithmetic (e.g. see \cite{Haase18}). See standard textbook in mathematical logic for more details on Presburger Arithmetic (e.g. see \cite{anderton-book}). Since Presburger Arithmetic admits quantifier-elimination, we may assume that $S$ is a boolean combination of (a) inequalities of linear combination of counting terms, and (b) modulo arithmetic on counting terms (i.e. an expression of the form $|w|_{a_i} \equiv k\pmod{c}$, for some concrete natural numbers $0 \leq k < c$ and $c > 0$).  For (b), one simply handles this using the formula $\Theta(\overrightarrow{\# a_i})$, where $\Theta$ is a unary numerical predicate consisting of all numbers $n$ such that $n \equiv k\pmod{c}$. For (a), take a linear inequality of the form
\[
    \psi(|w|_{a_1},\ldots,|w|_{a_d}) := \sum_{i=1}^d c_i |w|_{a_i} \geq 0,
\]
where $c_1,\ldots,c_d \in \mathbb{Z}$. Such a formula $\psi$ is already an atom permitted in $\LTLCP$. Since $\LTLCP$ is closed under boolean combination, it follows that $perm(L)$ is also in $\LTLCP$ and therefore, by Theorem \ref{thm:terms}, is in AHAT.
\end{proof}